\documentclass[letterpaper, 10 pt, conference]{ieeeconf}

\IEEEoverridecommandlockouts                              

\usepackage{graphicx}      
\usepackage{amssymb,amsmath}
\usepackage{color,soul}
\usepackage{subfigure}
\usepackage{epstopdf}   

\newtheorem{lem}{Lemma}
\newtheorem{remk}{Remark}

\newtheorem{problem}{Problem}

\usepackage{float}
\usepackage{lmodern}
\usepackage{subfig}
\usepackage{mathtools}
\usepackage{url}

\usepackage{boldline} 

\newcommand{\norm}[1]{\left\lVert#1\right\rVert}

\DeclarePairedDelimiter\floor{\lfloor}{\rfloor}

\DeclareMathAlphabet{\pazocal}{OMS}{zplm}{m}{n}

\title{\LARGE \bf
Design of input for data-driven simulation \\ with Hankel and Page matrices}
\author{Andrea Iannelli, Mingzhou Yin, Roy S. Smith
\thanks{
This work is supported by the Swiss National Science Foundation under grant no. $200021\_178890$.
\newline
The authors are with the Department of Electrical Engineering, Automatic Control Lab, ETH, Z\"{u}rich 8092, Switzerland
{\tt\small iannelli/myin/rsmith@control.ee.ethz.ch}}
}

\begin{document}

\maketitle
\thispagestyle{empty}
\pagestyle{empty}

\begin{abstract}
The paper deals with the problem of designing informative input trajectories for data-driven simulation. First, the excitation requirements in the case of noise-free data are discussed and new weaker conditions, which assume the simulated input to be known in advance, are provided. Then, the case of noisy data trajectories is considered and an input design problem based on a recently proposed maximum likelihood estimator is formulated. A Bayesian interpretation is provided, and the implications of using Hankel and Page matrix representations are demonstrated. Numerical examples show the impact of the designed input on the predictive accuracy.
\end{abstract}

\section{Introduction}

Predicting the response of a plant to given initial conditions and input signals is a fundamental task for analysis and control of dynamical systems. While this is a basic problem when a model of the plant is available (either derived from first principles or identified from experiments), increasing interest has been devoted to \emph{direct} data-driven prediction methods, whereby future responses are expressed in terms of past (data) trajectories. 
The behavioral approach \cite{Polderman_Willems_book} has found undisputed success in recent years within this context. Building on the seminal work \cite{Willems_2005} providing conditions under which the subspace of trajectories of a linear system can be spanned by noise-free data matrices, extensive research has been done in this direction, especially with the goal of using this non-parametric description for control \cite{Markovsky2021_review}.

As the success of such data-driven controllers highly depends on
the accuracy of the predictions, we focus here on the data-driven simulation problem for linear systems and consider two interrelated questions. The first, to which Section \ref{Background} is devoted, is concerned with characterizing the excitation requirements of the data trajectory. This problem was originally addressed in \cite{Markovsky_2008} by assuming persistence of excitation of the input and then applying the result from \cite{Willems_2005}. It is shown in this paper that, for a particular simulation task, the requirements can be greatly relaxed. While these results guarantee that the simulated output can be exactly reconstructed, they all require clean (or noise-free) data trajectories. Building on these weaker excitation requirements, the second contribution, presented in Section \ref{ExpDes}, is a design procedure to choose the input data trajectory in order to maximize the accuracy of the simulated response when the data trajectory is contaminated by noise. The accuracy criterion is formulated from a Bayesian viewpoint by leveraging a maximum likelihood data-driven estimator \cite{Ming_SMM_Arxiv} and the concept of mutual information \cite{Cover_inf_theory}.

We put emphasis, in both results, on the effect of using Hankel and Page matrix representations for the data trajectory. While Hankel matrices have a rich history in system identification \cite{van1996subspace}, Page matrices, proposed in \cite{Damen_Page_2982SCL} in the context of realization algorithms with noisy Markov parameters, have received less attention. It was shown in \cite{Damen_Page_2982SCL} that in the Page matrix case, since there are no repeated entries, de-noising by thresholding the lowest singular values is provably optimal. This is not the case for Hankel matrices where low-rank approximations via the SVD implicitly attribute a non-uniform weight on the noisy Markov parameters. The advantages pertaining to the use of Page matrices as distributionally robust predictors were recently discussed in \cite{coulson2020distributionally_TAC}. However, in the behavioural setting Page matrices are reportedly less accurate than Hankel ones, when the comparison is done with equal data length, owing to their less favourable sample efficiency \cite{Markovsky2020identifiability}. Numerical results in Section \ref{Results} show that, by using the proposed input design strategy, this aspect can be ameliorated and Page matrices do offer better predictive accuracy, as conceptually expected as they eliminate noise coupling among entries.

The input design problem in the data-driven setting has recently started receiving attention \cite{vanWaarde_expDes,Iannelli_SYSID_2021}. In \cite{vanWaarde_expDes}, an online procedure to select an input sequence which satisfies the conditions in \cite{Willems_2005} without requiring persistence of excitation is proposed in the case of clean data. The work in \cite{Iannelli_SYSID_2021} considered the case of impulse response simulation with noisy Hankel matrices, and used the standard excitation requirements \cite{Markovsky_2008}.

\subsection*{Notation and definitions}
The mutual information \cite{Cover_inf_theory}  of two multivariate random variables $x$ and $y$ of size $n$ is defined as
       \begin{equation}\label{mutual_inf}
I(x;y)=H(x)+H(y)-H(x,y)=H(x)-H(x|y),
    \end{equation}
where $H(x)=-\int p(x) \log(p(x))dx$ is the Shannon's entropy of $x$. If $x$ has Gaussian distribution with covariance $\Sigma_{\text{x}}$, it holds $H(x)=\frac{1}{2}\log(2\pi e)^n+\frac{1}{2}\log(\det(\Sigma_{\text{x}}))$.

Given a  matrix $X \in \mathbb{R}^{n \times n}$, $X^{\dagger}$ denotes its pseudo-inverse 
and $\{ \lambda_{i}(X), i=1,...,n\}$ the set of its eigenvalues. The symbol $I_n$ denotes the identity matrix of size $n$.

Given a signal $z \in \mathbb{R}^{n_z}$, we use $z_{[i,j]}$ to denote both the concatenated vector
$\left[ z_i^{\top}\;  ... \; z_j^{\top} \right]^{\top}$ and the sequence $ \{z_k\}_{k=i}^{j}$ of length $N=j-i+1$.
Given $z_{[i,j]}$, the associated block Hankel and Page matrices with $L$ block rows are defined respectively as
\begin{equation*}
\begin{aligned}
\mathcal{H}_{L}(z_{[i,j]})&:=\begin{bmatrix}
    z_{i} & z_{i+1} & \cdots & z_{N+i-L} \\
    z_{i+1}& z_{i+2} & \cdots & z_{N+i-L+1} \\
    \vdots & \vdots &  & \vdots \\
    z_{i+L-1} &  z_{i+L}  & \cdots   & z_{N+i-1}\\
    \end{bmatrix}, \\ 
\mathcal{P}_{L}(z_{[i,j]})&:=\begin{bmatrix}
    z_{i} & z_{i+L} & \cdots & z_{i+\floor*{\frac{N}{L}}L-L} \\
    z_{i+1}& z_{i+L+1} & \cdots & z_{i+\floor*{\frac{N}{L}}L-L+1} \\
    \vdots & \vdots &  & \vdots \\
    z_{i+L-1} &  z_{i+2L-1}  & \cdots   & z_{i+\floor*{\frac{N}{L}}L-1}\\
    \end{bmatrix}. \\ 
\end{aligned}
\end{equation*}
where $\mathcal{H}_{L}(z_{[i,j]}) \in \mathbb{R}^{L n_z \times (N-L+1)}$ and $\mathcal{P}_{L}(z_{[i,j]}) \in \mathbb{R}^{L n_z \times \floor*{\frac{N}{L}}}$.
The signal $z_{[i,j]}$ is persistently exciting (PE) of order $L$ if $\mathcal{H}_{L}(z_{[i,j]})$ has full row rank \cite{Willems_2005} and $L$-Page exciting of order $M$ if the matrix
$\begin{bmatrix}
\mathcal{P}_{L}(z_{[i,i+\floor*{\frac{N}{L}}L-1-(M-1)L]})   \\ 
\mathcal{P}_{L}(z_{[L+i,i+\floor*{\frac{N}{L}}L-1-(M-2)L]})   \\ 
...   \\ 
\mathcal{P}_{L}(z_{[L(M-1)+i,i+\floor*{\frac{N}{L}}L-1]})   \\ 
\end{bmatrix}$
has full row rank \cite{coulson2020distributionally_TAC}.

Given a state-space model $(A,B,C,D)$, the block Toeplitz matrix of impulse response coefficients, the extended observability matrix, and the reversed extended
controllability matrix are defined respectively as
\begin{equation*}
\begin{aligned}
 \mathcal{T}_{i}&=\begin{bmatrix} 
    D & 0 & 0 & 0 \\
    CB& D & 0 & 0 \\
    \vdots & \ddots & D & 0 \\
     C A^{i-2} B & C A^{i-3} B & \cdots & D \\
    \end{bmatrix}, \\
\mathcal{O}_i&=\begin{bmatrix}C^{\top}\; \; (CA)^{\top}\;  ... \; (CA^{i-1})^{\top} \end{bmatrix}^{\top},\\
\mathcal{C}_i&=\begin{bmatrix}A^{i-1}B\; ... \; AB\; \; B \end{bmatrix}.\\
\end{aligned}
\end{equation*}


\section{Data-driven simulation}\label{Background}

\subsection{Problem setting and available results}\label{DD_std}

Consider a linear time-invariant (LTI) system
\begin{subequations}\label{system}
\begin{align}
 x_{t+1}&=A x_{t}+B u_{t}, \label{system_x}\\
 y_{t}&=C x_{t}+D u_{t}, \label{system_y} 
\end{align}
\end{subequations}
where $x \in \mathbb{R}^{n_x}$ is the state, $u \in \mathbb{R}^{n_u}$ is the input, and $y \in \mathbb{R}^{n_y}$ is the output. It is assumed that (\ref{system}) is minimal, i.e. the system is controllable, observable, and $n_x$ is its McMillan degree.
The lag $l$ is defined as the smallest integer $i$ such that $\mathcal{O}_i$ has rank $n_x$.
The data-driven (or model-free) simulation problem for (\ref{system}) is stated next.
\begin{problem}\label{DDsim_problem_ISO}
Given an input-output \emph{data} trajectory $(u_{\text{d} \; [0,N-1]},y_{\text{d}  \; [0,N-1]})$, an input \emph{simulation} trajectory $u_{\text{s} \; [0,L_s-1]}$,
and an initial condition $x_{\text{ini}}$, find the (unique) output \emph{simulation} trajectory $y_{\text{s} \; [0,L_s-1]}$.
\end{problem}
It was shown in \cite{Markovsky_2008} that Problem \ref{DDsim_problem_ISO} can be solved in a fully input-output (or representation-free) setting by leveraging results from behavioral system theory \cite{Polderman_Willems_book}. The first observation is that an input-output \emph{initial} trajectory $(u_{\text{ini} \; [0,L_0-1]},y_{\text{ini} \; [0,L_0-1]})$ of (\ref{system}) can be used, if $L_0 \geq l$, to uniquely define the initial condition $x_{\text{ini}}$ of Problem \ref{DDsim_problem_ISO}. Define $L=L_0+L_s$. 
\begin{lem}\label{DDsim_Markovsky_IJC}\cite{Markovsky_2008} 
Given $(u_{\text{ini} \; [0,L_0-1]},y_{\text{ini} \; [0,L_0-1]})$, with $L_0 \geq l$, $(u_{\text{d} \; [0,N-1]},y_{\text{d}  \; [0,N-1]})$, and $u_{\text{s} \; [0,L_s-1]}$.
Assume that the data generating system is controllable and that $u_{\text{d} \; [0,N-1]}$ is persistently exciting of order $L+n_x$. Partition Hankel matrices built using the data trajectory according to the indices $L_0$ (p) and $L_s$ (f) as follows 
       \begin{equation}\label{partition_H}
    \begin{bmatrix}U\\ \hline Y \end{bmatrix}
   = \begin{bmatrix}U_{\text{p}}\\U_{\text{f}}\\ \hline Y_{\text{p}}\\Y_{\text{f}} \end{bmatrix}
    :=\begin{bmatrix}\mathcal{H}_{L}(u_{\text{d} \; [0,N-1]})\\ \hline \mathcal{H}_{L}(y_{\text{d} \; [0,N-1]})\end{bmatrix}.
    \end{equation}
Then, $y_{\text{s} \; [0,L_s-1]}= Y_{\text{f}} g$, where $g$ satisfies
\begin{equation}\label{DDsim_Hankel_eq_g}
\begin{bmatrix}
u_{\text{ini} \; [0,L_0-1]}   \\ 
y_{\text{ini} \; [0,L_0-1]} \\
u_{\text{s} \; [0,L_s-1]}  \\
\end{bmatrix}=\begin{bmatrix}
U_{\text{p}}\\ 
Y_{\text{p}} \\
U_{\text{f}}  \\
\end{bmatrix}g.
\end{equation}
\end{lem}

The proof is an application of the so-called \emph{Fundamental Lemma} (FL) \cite{Willems_2005}, which guarantees, under the assumptions of Lemma \ref{DDsim_Markovsky_IJC}, that any trajectory generated by (\ref{system}) is spanned by $\begin{bmatrix}\mathcal{H}_{L}(u_{\text{d} \; [0,N-1]})\\ \hline \mathcal{H}_{L}(y_{\text{d} \; [0,N-1]})\end{bmatrix}$. Uniqueness of the output \emph{simulation} trajectory is guaranteed by the fact that $x_{\text{ini}}$ is uniquely defined by an input-output trajectory of appropriate length.
Lemma \ref{DDsim_Markovsky_IJC} thus provides an answer to Problem \ref{DDsim_problem_ISO} by virtue of Hankel data matrices that characterize the system's behavior.
It has been recently shown in \cite{coulson2020distributionally_TAC} that Page data matrix representations can be used as well to span the system's trajectories. A straightforward application of the results from \cite{coulson2020distributionally_TAC} to the same arguments used in \cite{Markovsky_2008} yields the following result on data-driven simulation with Page matrices.
\begin{lem}\label{DDsim_Markovsky_Page}
Given $(u_{\text{ini} \; [0,L_0-1]},y_{\text{ini} \; [0,L_0-1]})$, with $L_0 \geq l$, $(u_{\text{d} \; [0,N-1]},y_{\text{d}  \; [0,N-1]})$, and $u_{\text{s} \; [0,L_s-1]}$.
Assume that the data generating system is controllable and that $u_{\text{d} \; [0,N-1]}$ is $L$-Page exciting of order $n_x+1$. Partition Page matrices built using the data trajectory according to the indices $L_0$ (p) and $L_s$ (f) as follows 
       \begin{equation}\label{partition_P}
    \begin{bmatrix}U\\ \hline Y \end{bmatrix}
    =\begin{bmatrix}U_{\text{p}}\\U_{\text{f}}\\ \hline Y_{\text{p}}\\Y_{\text{f}} \end{bmatrix}
    :=\begin{bmatrix}\mathcal{P}_{L}(u_{\text{d} \; [0,N-1]})\\ \hline \mathcal{P}_{L}(y_{\text{d} \; [0,N-1]})\end{bmatrix}.
    \end{equation}
Then, $y_{\text{s} \; [0,L_s-1]}= Y_{\text{f}} g$, where $g$ satisfies
\begin{equation}\label{DDsim_Page_eq_g}
\begin{bmatrix}
u_{\text{ini} \; [0,L_0-1]}   \\ 
y_{\text{ini} \; [0,L_0-1]} \\
u_{\text{s} \; [0,L_s-1]}  \\
\end{bmatrix}=\begin{bmatrix}
U_{\text{p}}\\ 
Y_{\text{p}} \\
U_{\text{f}}  \\
\end{bmatrix}g.
\end{equation}
\end{lem}

While, by uniqueness, they provide the same output simulation trajectory, the Hankel and Page matrices formulations are markedly different as far as data length requirements are concerned, because of the different excitation requirements. To guarantee the conditions of Lemma \ref{DDsim_Markovsky_IJC}, it must hold $N\geq (L+n_x)(n_u+1)-1$. Instead, one needs $N\geq L((n_u L +1)(n_x+1)-1)$ to apply Lemma \ref{DDsim_Markovsky_Page}. This of course implies that the Page formulation typically requires a much longer data set, with potential negative effects on experimental and computational costs. Another more subtle consequence is related to the fact that, for a given $N$, the number of columns $c_H$ of a Hankel matrix will always be larger than the number of columns $c_P$ of a Page matrix. Specifically, $L(c_H(N)-c_P(N))=N(L-1)-L^2+L$, thus this gap increases linearly with $N$. The implication is that, in the range of $N$ for which the two formulations can be compared (i.e. $N\geq L((n_u L +1)(n_x+1)-1)$), $c_H \gg c_P$. This represents a strong disadvantage of the Page matrix when working with noisy data trajectories, since the effect of noise can be averaged by increasing the number of columns of the data matrices. This fact plays a decisive role when comparing Hankel and Page predictive performance with noisy data, as recently observed in \cite{Markovsky2020identifiability}.

\subsection{Weaker excitation conditions}\label{DD_relax}

The excitation conditions prescribed by Lemmas \ref{DDsim_Markovsky_IJC} and \ref{DDsim_Markovsky_Page} are required because the respective data matrices are used to solve the simulation problem for \emph{any} input simulation trajectory and initial condition. However, when these are known in advance, the requirements on the data trajectory can be significantly relaxed, as shown next.

As a preamble, it is recalled that a necessary and sufficient condition for a generic input-output trajectory $(u_{[0,T-1]},y_{[0,T-1]})$ to be generated by (\ref{system}) is that there exists $x_0 \in  \mathbb{R}^{n_x}$ such that
\begin{equation}\label{iff_uy_B}
y_{[0,T-1]}=\mathcal{O}_{T} x_0+\mathcal{T}_{T} u_{[0,T-1]},
\end{equation}
as can be shown by writing (\ref{system_y}) for $t$ $\in$ $[0, T-1]$. 
Since by assumption $(u_{\text{ini} \; [0,L_0-1]},y_{\text{ini} \; [0,L_0-1]})$ is a trajectory of (\ref{system}), 
it satisfies (\ref{iff_uy_B}), and we denote by $\hat{x}$ the initial condition of the state for which (\ref{iff_uy_B}) is verified.
We also denote by $x_{\text{d} \; [0,N-1]}$ the state trajectory associated with the input-output data trajectory.
\begin{lem}\label{DDsim_Relaxed}
Given $(u_{\text{ini} \; [0,L_0-1]},y_{\text{ini} \; [0,L_0-1]})$, with $L_0 \geq l$, $(u_{\text{d} \; [0,N-1]},y_{\text{d}  \; [0,N-1]})$, and $u_{\text{s} \; [0,L_s-1]}$. Partition Hankel matrices as in (\ref{partition_H}). Assume that
\begin{equation}\label{DDsim_Relaxed_A12}
\begin{bmatrix}
u_{\text{ini} \; [0,L_0-1]}   \\
u_{\text{s} \; [0,L_s-1]}  \\
\hline
\hat{x} \\
\end{bmatrix}\in \text{Im}\left( \begin{bmatrix}
U_{\text{p}}\\
U_{\text{f}}  \\
\hline
X_{\text{p}} \\
\end{bmatrix}\right),
\end{equation}
where $X_{\text{p}}=[x_{\text{d},0} \; x_{\text{d},1} \;... \; x_{\text{d},N-L}]$.

Then, $y_{\text{s} \; [0,L_s-1]}= Y_{\text{f}} g$, where $g$ satisfies
\begin{equation}\label{DDsim_Relaxed_eq_g}
\begin{bmatrix}
u_{\text{ini} \; [0,L_0-1]}   \\ 
y_{\text{ini} \; [0,L_0-1]} \\
u_{\text{s} \; [0,L_s-1]}  \\
\end{bmatrix}=\begin{bmatrix}
U_{\text{p}}\\ 
Y_{\text{p}} \\
U_{\text{f}}  \\
\end{bmatrix}g.
\end{equation}
\end{lem}
\begin{proof}
Specializing (\ref{iff_uy_B}) to the initial trajectory gives
\begin{subequations}\label{iff_uy_ini}
\begin{align}
&y_{\text{ini} \; [0,L_0-1]}=\mathcal{O}_{L_0} \hat{x}+\mathcal{T}_{L_0} u_{\text{ini} \; [0,L_0-1]}, \label{iff_uy_ini_1}\\
&\hat{x} =\mathcal{O}_{L_0}^{\dagger} y_{\text{ini} \; [0,L_0-1]} - \mathcal{O}_{L_0}^{\dagger} \mathcal{T}_{L_0} u_{\text{ini} \; [0,L_0-1]},\label{iff_uy_ini_2}
\end{align}
\end{subequations}
where the vector $\hat{x}$ is unique, since by assumption $L_0 \geq l$ and thus $\mathcal{O}_{L_0}$ has full rank. Using (\ref{DDsim_Relaxed_A12}), we conclude from (\ref{iff_uy_ini}) that
$\begin{bmatrix}
u_{\text{ini} \; [0,L_0-1]}   \\
y_{\text{ini} \; [0,L_0-1]}  \\
u_{\text{s} \; [0,L_s-1]}  \\
\end{bmatrix}\in \text{Im}\left( \begin{bmatrix}
U_{\text{p}}\\
\mathcal{O}_{L_0} X_{\text{p}} + \mathcal{T}_{L_0} U_{\text{p}} \\
U_{\text{f}}\\
\end{bmatrix}\right)$.
It can be shown starting from (\ref{system}), see e.g. the subspace identification literature \cite{van1996subspace}, that
the data trajectory $(u_{\text{d} \; [0,N-1]},y_{\text{d}  \; [0,N-1]})$ satisfies the matrix equation \\ $Y_{\text{p}}=\mathcal{O}_{L_0} X_{\text{p}} + \mathcal{T}_{L_0} U_{\text{p}}$. Therefore
\begin{equation}\label{range_cond}
\begin{bmatrix}
u_{\text{ini} \; [0,L_0-1]}   \\
y_{\text{ini} \; [0,L_0-1]}  \\
u_{\text{s} \; [0,L_s-1]}  \\
\end{bmatrix}\in \text{Im}\left( \begin{bmatrix}
U_{\text{p}}\\
Y_{\text{p}}\\
U_{\text{f}}\\
\end{bmatrix}\right).
\end{equation}

Consider now (\ref{iff_uy_B}) for the simulation trajectory
\begin{equation}\label{iff_uy_s}
y_{\text{s} \; [0,L_s-1]}=\mathcal{O}_{L_s} x_{\text{ini}}+\mathcal{T}_{L_s} u_{\text{s} \; [0,L_s-1]},
\end{equation}
where, to obtain $x_{\text{ini}}$, write (\ref{system_x}) in the interval $[0,L_0]$
\begin{equation}\label{x_ini_expl}
x_{\text{ini}}=x_{L_0}=A^{L_0}\hat{x}+\mathcal{C}_{L_0}u_{\text{ini} \; [0,L_0-1]}.
\end{equation}
Because $\hat{x}$ is unique, $x_{\text{ini}}$ is uniquely defined. Substituting the explicit expression of $\hat{x}$ (\ref{iff_uy_ini_2}) in (\ref{x_ini_expl}) yields  
\begin{equation*}
x_{\text{ini}}=\underbrace{\left(\mathcal{C}_{L_0}- A^{L_0}\mathcal{O}_{L_0}^{\dagger} \mathcal{T}_{L_0}\right)}_\text{$P$} u_{\text{ini} \; [0,L_0-1]} + \underbrace{A^{L_0}\mathcal{O}_{L_0}^{\dagger}}_\text{$Q$} y_{\text{ini} \; [0,L_0-1]}.
\end{equation*}
Rewrite (\ref{iff_uy_s}) equivalently as
\begin{equation}\label{iff_uy_s_PQ}
\begin{aligned}
y_{\text{s} \; [0,L_s-1]}&=\begin{bmatrix}\mathcal{O}_{L_s} P \; \; \mathcal{O}_{L_s} Q  \; \;  \mathcal{T}_{L_s}\end{bmatrix}
\begin{bmatrix}
u_{\text{ini} \; [0,L_0-1]}   \\ 
y_{\text{ini} \; [0,L_0-1]} \\
u_{\text{s} \; [0,L_s-1]}  \\
\end{bmatrix},\\
&=\underbrace{\begin{bmatrix}\mathcal{O}_{L_s} P \; \; \mathcal{O}_{L_s} Q  \; \;  \mathcal{T}_{L_s}\end{bmatrix}
\begin{bmatrix}
U_{\text{p}}\\ 
Y_{\text{p}} \\
U_{\text{f}}  \\
\end{bmatrix}}_\text{$K$}g,\\
\end{aligned}
\end{equation}
where the existence of the vector $g$ in the last equality is guaranteed by (\ref{range_cond}).
Because (\ref{iff_uy_s}) uniquely defines $y_{\text{s} \; [0,L_s-1]}$, showing $K=Y_f$ proves the statement.
To see this, note that writing (\ref{system}) in matrix form for the data trajectory $(u_{\text{d} \; [0,N-1]},y_{\text{d}  \; [0,N-1]})$ also yields $Y_{\text{f}}=\mathcal{O}_{L_s} X_{\text{f}} + \mathcal{T}_{L_s} U_{\text{f}}$, where $X_{\text{f}}=[x_{\text{d},L_0} \; x_{\text{d},L_0+1}... \; x_{\text{d},L_0+N-L}]$.
Note also that, by their definition, $P$ and $Q$ predict the current state given previous input and output sequences. This of course also holds for the data trajectory, and thus we have that
$X_{\text{f}}=\begin{bmatrix} P \; \; Q  \end{bmatrix}\begin{bmatrix}
U_{p}\\ 
Y_p \\
\end{bmatrix}$. Therefore
\begin{equation}
K=\mathcal{O}_{L_s} \begin{bmatrix} P \; \; Q  \end{bmatrix}\begin{bmatrix}
U_{\text{p}}\\ 
Y_{\text{p}} \\
\end{bmatrix}+\mathcal{T}_{L_s}U_{\text{f}}
=Y_{\text{f}},
\end{equation}
which concludes the proof.
\end{proof}

The result is proved for Hankel data matrices, but it can be seamlessly extended to Page matrices by simply changing the notation. That is, define $X_{\text{p}}=[x_{\text{d}, \;0} \; x_{\text{d}, \;L}... \; x_{\text{d}, \;(\floor*{\frac{N}{L}}-1)L}]$ and $X_{\text{f}}=[x_{\text{d}, \;L_0} \; x_{\text{d}, \;L_0+L}... \; x_{\text{d}, \;L_0+(\floor*{\frac{N}{L}}-1)L}]$, and replace partition (\ref{partition_H}) by (\ref{partition_P}). Notably, the excitation requirements do not change between the two data representations, in contrast to the previous results discussed in Section \ref{DD_std}.

\begin{remk}
Equation \ref{DDsim_Relaxed_A12} can in principle be satisfied with rank 1 data matrices, i.e. the minimum number of columns is 1 and the minimum length $N$ of the data trajectory is $L$. In this case, the range condition on $\hat{x}$ is satisfied if and only if the initial conditions of the initial and data trajectories coincide up to a scaling factor, while for $N>L$ this is only a sufficient condition. Note also that the assumption on $\hat{x}$ in (\ref{DDsim_Relaxed_A12}) is only required to guarantee (\ref{range_cond}). The latter is a condition that can be checked directly from the data, and can thus be used in lieu of (\ref{DDsim_Relaxed_A12}) to define the set of $y_{\text{ini} \; [0,L_0-1]}$ which can be simulated given an input trajectory satisfying the excitation conditions in (\ref{DDsim_Relaxed_A12}).
For controllable systems this is in principle w.l.o.g. because, as also shown in (\ref{x_ini_expl}), any $x_{\text{ini}}$ can be generated by $u_{\text{ini} \; [0,L_0-1]}$ alone.
\end{remk}

\section{Input design for simulation with noisy data}\label{ExpDes}

The main question to be addressed is: What is the best input data trajectory $u_{\text{d} \; [0,N-1]}$ to solve Problem 1? Lemmas \ref{DDsim_Markovsky_IJC}, \ref{DDsim_Markovsky_Page}, and \ref{DDsim_Relaxed} are already experiment design results, as also recognized in \cite{vanWaarde_expDes} with respect to the Fundamental Lemma, because they provide both the minimum experiment length and the required signal properties by characterizing the span of the associated data matrices. In fact, when the initial condition and the input simulation trajectory are set in advance, Lemma \ref{DDsim_Relaxed} gives the less restrictive excitation requirements. However, all of these results rely on the assumption of working with clean data. If the data trajectory is contaminated with noise, computing exactly the output simulation trajectory is not possible as these results no longer hold. A possible strategy to make use of these methodologies while reducing the effect of the noise was presented in \cite{Ming_SMM_Arxiv}, under the name of signal matrix model (SMM). This is reviewed next, with an emphasis on some novel aspects, since it will be used later to frame the input design problem.

For the sake of readability, the time indexes are dropped from the sequences' subscripts. The same time indexes employed in the previous section apply for the respective trajectories (e.g. $y_{\text{s}}$ will denote $y_{\text{s} \; [0,L_s-1]}$).

\subsection{The SMM estimator}\label{ExpDes_SMM}
The signal matrix model is a maximum likelihood estimator of the vector $g$ introduced in the previous lemmas. This was proposed in \cite{Ming_SMM_Arxiv} in conjunction with Lemma \ref{DDsim_Markovsky_IJC}, and, under its assumptions and Hankel data matrices partitions, provided a predicted output simulation trajectory $y_{\text{s}}= Y_{\text{f}} g_{\text{SMM}}$ with favourable statistical properties.

In the same spirit, a maximum likelihood estimator $g_{\text{SMM}}$ can be defined building on Lemma \ref{DDsim_Relaxed}. We consider the case where the output data trajectory ($n_y=1$ is assumed here for simplicity of representation) is subject to i.i.d. Gaussian noise
\begin{equation}\label{noise_data}
\tilde{y}_{\text{d},i}= y_{\text{d},i}+w_i,\; \; w_i\sim \pazocal{N}(0,\sigma^2),\; \; i=1,...,N-1.
\end{equation}
The rest of the problem's data ($u_{\text{d}}$, $u_{\text{ini}}$, and $y_{\text{ini}}$) are assumed to be noise-free because either they will be optimized over later (i.e. the input data trajectory) or are fixed by the analyst (i.e. the initial trajectory).

\begin{lem}\label{SMM_Relaxed}
Given $(u_{\text{ini}},y_{\text{ini}})$, with $L_0 \geq l$, $(u_{\text{d}},\tilde{y}_{\text{d}})$, and $u_{\text{s}}$ satisfying the assumptions of Lemma \ref{DDsim_Relaxed}. Partition the Hankel matrices as in (\ref{partition_H}). Define the random variable $\bar{y}:=Y g-\left[\begin{smallmatrix}y_{\text{ini}}\\ 0 \end{smallmatrix}\right]$ where the usual partition applies. The value of $g$ that maximizes $ \mathbb{E}(\bar{y}|g)$, i.e. the conditional probability of observing the realization $\bar{y}$ corresponding to the available data given $g$, is
given by
\begin{equation}
\underset{g\in\pazocal{G}}{\text{min}}\
    \text{logdet}(\Sigma_{\text{y}}(g))+\begin{bmatrix}Y_\text{p} g-y_{\text{ini}}\\ 0\end{bmatrix}^{\top}\Sigma_{\text{y}}^{-1}(g)\begin{bmatrix}Y_\text{p} g-y_{\text{ini}}\\ 0\end{bmatrix},
\label{eqn:opt0}
\end{equation}
where $\pazocal{G}$ is the set
\begin{equation}
    \pazocal{G} = \left\{g\in \mathbb{R}^{N-L+1}\left|\begin{bmatrix}
    U_\text{p}\\U_\text{f}
    \end{bmatrix}g=\begin{bmatrix}
    u_{\text{ini}} \\ u_{\text{s}}
    \end{bmatrix}\right.\right\},
\end{equation}
and
\begin{equation}    \label{eqn:stats}
    \left(\Sigma_{\text{y}}\right)_{i,j} = \left( \text{cov}(\bar{y}|g)\right)_{i,j}= \sigma^2\sum_{k=1}^{N-L+1-|i-j|}g_k g_{k+|i-j|}.
\end{equation}
The maximum likelihood simulation trajectory is then 
\begin{equation}\label{SMM_sim}
\hat{y}_{\text{s},\text{\tiny SMM}}= Y_{\text{f}} g_{\text{\tiny SMM}},
\end{equation}
where $g_{\text{SMM}} \in \arg \min (\ref{eqn:opt0})$.
\end{lem}

Lemma \ref{SMM_Relaxed} combines the maximum likelihood formulation from \cite{Ming_SMM_Arxiv}, which leads to problem (\ref{eqn:opt0}), with the result in Lemma \ref{DDsim_Relaxed}. The SMM estimator can also be postulated for Page matrix representations. 
Besides the notational difference of replacing partition (\ref{partition_H}) by (\ref{partition_P}), an important distinction is that for Page matrices
\begin{equation*}
    \Sigma_{\text{y}} = \text{cov}(\bar{y}|g)= \sigma^2 \norm{g}_2^2 I_{L}.
\end{equation*}
That is, the covariance matrix appearing in the optimization problem (\ref{eqn:opt0}) is diagonal. 
This comes from the fact that $\text{cov}(\bar{y}|g)=\left(g^\top \otimes I_L\right)\text{cov}\left(\text{vec}(Y)\right)\left(g\otimes I_L \right)$.
Due to the absence of repeated entries in the Page matrix (\ref{partition_P}), $\text{vec}(Y)$ is a vector of uncorrelated random variables with covariance $\sigma^2 I_{N}$.
This is of course not the case when the structured Hankel matrix is used, due to the repetitions in each column, which leads to the banded structure in (\ref{eqn:stats}). Nonetheless, setting to zero the off-diagonal terms simplifies the solution of problem (\ref{eqn:opt0}), thus the approximation to a diagonal $\Sigma_{\text{y}}$ was proposed in \cite{Ming_SMM_Arxiv} even when working with Hankel matrices.

\subsection{Input design for SMM}\label{ExpDes_inpDes}

In \cite{Iannelli_SYSID_2021} the input design problem for identification of the truncated infinite impulse response using the SMM estimator with Hankel matrices was investigated.
The mean-square error (MSE) matrix \cite{LjungBook2} of the estimated response was chosen to measure the accuracy of the estimates. The input design problem was then formulated as the minimization of A-, D-, and E- optimality criteria. 
The main finding was that, if the off-diagonal entries of $\Sigma_{\text{y}}$ are neglected, minimizing these criteria is equivalent to minimizing the Euclidean norm of $g_{\text{SMM}}$.
As observed earlier, this is true only for Page matrices and not in general for Hankel matrices. In the latter case, minimization of the Euclidean norm of $g_{\text{SMM}}$ is justified from an A-optimality viewpoint, since this consists of minimizing the trace of the MSE matrix. Inspired by the recent work in \cite{Fujimoto_Auto_2018} concerning experiment design using tools from information theory, we provide here a Bayesian formulation of the SMM input design problem and show the implications of Page and Hankel matrix representations.

\subsubsection{Bayesian perspective}
The data-based SMM output simulation trajectory can be modelled as a Gaussian random variable (the subscript SMM will be dropped)  
\begin{equation}\label{Bayes_meas}
    y_{\text{s}}^{\text{D}} \sim \pazocal{N}(\hat{y}_{\text{s}},\Sigma_{\text{y,f}}),
\end{equation}
where $\hat{y}_{\text{s}}$ is given in (\ref{SMM_sim}) and $\Sigma_{\text{y,f}}$ is the matrix made of the last $L_{s}$ rows and columns of $\Sigma_{\text{y}}$ (evaluated at $g$). Assume that prior knowledge on $y_{\text{s}}$ is encoded in a positive definite kernel matrix $\Sigma_{\text{\tiny K}}$, resulting in a prior distribution $y^{0}_{\text{s}} \sim \pazocal{N}(0,\Sigma_{\text{K}})$. When the simulation problem consists of estimating the first $\text{$L_s$}$ coefficients of the truncated infinite impulse response \cite{Markovsky_2005}, one can refer to an extensive literature on choices of kernel matrices encoding priors related to system's theoretic properties \cite{Pillonetto_2010,Chen_Auto12}. However, the idea of using priors to improve on and regularize the estimate from data (\ref{Bayes_meas}) can in principle be used for other simulation problems as well. Examples of priors might include smoothness and decay rate of the response. We can then, in the spirit of Kalman filtering, combine prior and data-based estimates to provide the MSE estimate.
This can be interpreted as the posterior distribution of $y_{\text{s}}$ given the data trajectory
\begin{equation}
   \left( y_{\text{s}}| (u_{\text{d}},y_{\text{d}}) \right) \sim \pazocal{N}(K \hat{y}_{\text{s}},\Sigma_\text{post}),
\end{equation}
where $K=\Sigma_{\text{K}}\left(\Sigma_{\text{K}}+\Sigma_{\text{y,f}} \right)^{-1}$ is the Kalman gain and $\Sigma_\text{post}=\Sigma_{\text{K}}-\Sigma_{\text{K}}\left(\Sigma_{\text{K}}+\Sigma_{\text{y,f}}\right)^{-1}\Sigma_{\text{K}}$ is the posterior covariance. These expressions were already presented in \cite{Ming_SMM_Arxiv} and are standard filtering relationships \cite{Simon2006}. The novelty here is their interpretation in the data-driven simulation setting.
This is important as it enables the formal definition of the input design problem as the maximization of the distance between the prior and posterior distributions of $y_{\text{s}}$. By using an information theoretic result \cite{Cover_inf_theory}, the expected value of the KL divergence between these distributions coincides with the mutual information of $y_{\text{s}}$ and the data $(u_{\text{d}},y_{\text{d}})$.
Therefore, maximizing $I(y_{\text{s}};(u_{\text{d}},y_{\text{d}}))$ yields, from a Bayesian viewpoint, an informative experiment.

Using its definition (\ref{mutual_inf}), the mutual information for the case of interest can be defined as
       \begin{subequations}\label{mutual_inf_SMM}
       \begin{align}
&I(y_{\text{s}};(u_{\text{d}},y_{\text{d}}))=H(y_{\text{s}})-H(y_{\text{s}}| (u_{\text{d}},y_{\text{d}})), \nonumber\\
&= \frac{1}{2}\left(\log(\det(\Sigma_{\text{K}}))-\log(\det(\Sigma_{\text{post}}))\right) , \label{mutual_inf_SMM_2}\\
&= \frac{1}{2}\log(\det(I_{L_s}+\Sigma_{\text{K}}\Sigma_{\text{y,f}}^{-1})) , \label{mutual_inf_SMM_3}
    \end{align}
    \end{subequations}
where (\ref{mutual_inf_SMM_3}) comes from the fact that
       \begin{equation*}
       \begin{aligned}
\Sigma_\text{post}&=\Sigma_{\text{K}}-\left(\Sigma_{\text{K}}\Sigma_{\text{y,f}}^{-1}+I_{L_s}\right)^{-1}\Sigma_{\text{K}}\Sigma_{\text{y,f}}^{-1}\Sigma_{\text{K}},\\
&=\Sigma_{\text{K}}-\left(\Sigma_{\text{y,f}}^{-1}+\Sigma_{\text{K}}^{-1}\right)^{-1}\Sigma_{\text{y,f}}^{-1}\Sigma_{\text{K}},   \\
&=\left(\Sigma_{\text{y,f}}^{-1}+\Sigma_{\text{K}}^{-1}\right)^{-1}\left[ \left(\Sigma_{\text{y,f}}^{-1}+\Sigma_{\text{K}}^{-1}\right)\Sigma_{\text{K}}-\Sigma_{\text{y,f}}^{-1}\Sigma_{\text{K}} \right],\\
&= \left(\Sigma_{\text{y,f}}^{-1}+\Sigma_{\text{K}}^{-1} \right)^{-1}.  \\
    \end{aligned}
    \end{equation*}
The following result shows an important relationship between the mutual information and $g$.
\begin{lem}\label{ID_MutInf_lemma}
If $\Sigma_{\text{y,f}}= \sigma^2 \norm{g}_2^2 I_{L_s}$, then there exist functions $f$ and $h$ such that
\begin{equation*}
2I(y_{\text{s}};(u_{\text{d}},y_{\text{d}}))=\hat{I}(y_{\text{s}};(u_{\text{d}},y_{\text{d}}))=f(\norm{g}_2^2,\Sigma_{\text{K}})+h(\Sigma_{\text{K}}),
\end{equation*}
where $f(\cdot,\cdot)$ is monotonically decreasing with respect to the first argument irrespective of the second, and $h$ does not depend on $g$. 
\end{lem}
\begin{proof}
Substituting $\Sigma_{\text{y,f}}= \sigma^2 \norm{g}_2^2 I_{L_s}$ in (\ref{mutual_inf_SMM_2}), the (scaled) mutual information can be written as
       \begin{equation*}
       \begin{aligned}
&\hat{I}=\log\left(\det\left(\frac{1}{\sigma^2 \norm{g}_2^2} I_{L_s}+\Sigma_{\text{K}}^{-1}\right)\right) +\underbrace{\log(\det(\Sigma_{\text{K}}))}_\text{$h(\Sigma_{\text{K}})$},\\ 
&= \underbrace{-L_s \log(\sigma^2 \norm{g}_2^2)+\sum_{i=1}^{L_s}\log(1+\sigma^2 \norm{g}_2^2 \lambda_i(\Sigma_{\text{K}}^{-1}))}_\text{$f(\norm{g}_2^2,\Sigma_{\text{K}})$}+h(\Sigma_{\text{K}}).\\ 
    \end{aligned}
    \end{equation*}
where the definition of the characteristic polynomial is used. Define $z=\sigma^2 \norm{g}_2^2 \geq 0$. It can then be shown that 
\begin{equation*}
f(z,\Sigma_\text{K})=-L_s \log(z)+\sum_{i=1}^{L_s}\log(1+z \lambda_i(\Sigma_{\text{K}})),
\end{equation*}
is monotonically decreasing with respect to $z$, and thus $\norm{g}_2^2$, irrespective of $\Sigma_{\text{K}}$. Rewrite $f$ and its derivative as
\begin{equation*}
f=\log\left( \underbrace{\prod_{i=1}^{L_s} \left(\frac{1+z \lambda_i}{z} \right)}_\text{$p(z,\lambda_i)$} \right),\quad \frac{\partial f}{\partial z}=\frac{1}{p(z,\lambda_i)}\frac{\partial p(z,\lambda_i)}{\partial z},\\
\end{equation*}
and observe first that $\lambda_i(\Sigma_{\text{K}}^{-1})>0$ for all $i$, because they are eigenvalues of the inverse of a positive definite matrix.
Define $p_i(z,\lambda_i)=\left(\frac{1+z \lambda_i}{z} \right)$. Since $p_i(z,\lambda_i)>0$, it follows that $p(z,\lambda_i)>0$ and thus the derivatives of $f$ and $p$ with respect to $z$ have the same sign.  Monotonic decrease of $f$ with respect to $z$ then follows from the fact that
\begin{equation*}
\begin{aligned}
\frac{\partial p_i(z,\lambda_i)}{\partial z}&=-\frac{1}{z^2},\\
\frac{\partial p(z,\lambda_i)}{\partial z}&= \underbrace{\left(\prod_{i=1}^{L_s}p_i(z,\lambda_i)\right)}_\text{$\geq0$}\underbrace{\left(\sum_{j=1}^{L_s}\frac{\partial p_j}{\partial z}\frac{ 1}{p_j}\right)}_\text{$\leq0$}.\\
\end{aligned}
\end{equation*}
\end{proof}

Recall that, when Page matrices are used, $\Sigma_{\text{y,f}}$ has the diagonal structure assumed in the lemma. This has two important implications for the input design problem.
First, searching for the input sequence which minimizes the Euclidean norm of $g$ also maximizes the mutual information. Second, the prior on the simulation output, introduced via the kernel $\Sigma_{\text{K}}$, has no effect on the input design problem. That is, maximizing the mutual information coincides with minimizing uncertainty in the data estimate (\ref{Bayes_meas}). Precisely, it coincides with the D-optimality criterion applied to $\Sigma_{\text{y,f}}$.
This is in contrast with recent results on input design for kernel-based impulse response identification \cite{Fujimoto_Auto_2018}, and is a property of the SMM estimator when the Page matrix representation is used.

\subsubsection{Optimization problem}

Motivated by the Bayesian interpretation, the experiment design optimization problem can be defined, at an abstract level, as the solution to the bi-level optimization problem
\begin{subequations}\label{ED_prog_bi-lev}
\begin{align}
    \underset{g,u_{\text{d}}}{\text{min}}&  \norm{g}_2^2, \label{ED_prog_bi-lev_1}\\
\text{s.t.}& \; \; u_{\text{d}} \in \mathcal{U}, \label{ED_prog_bi-lev_3} \\
    & g \in \underset{g\in\pazocal{G},u_{\text{d}} }{\arg \min } (\ref{eqn:opt0}), \label{ED_prog_bi-lev_2}
\end{align}
\end{subequations}
where $\mathcal{U}$ defines input constraints. Using the relaxations of the SMM objective function (\ref{eqn:opt0}) suggested in \cite{Ming_SMM_Arxiv}, and the idea to replace the inner optimization problem by its KKT conditions used in \cite{Iannelli_SYSID_2021}, problem (\ref{ED_prog_bi-lev}) can be formulated as the following nonlinear program
\begin{subequations}\label{ED_prog}
\begin{align}
&    \underset{g,u_{\text{d}},\nu}{\text{min}}  \norm{g}_2^2, \label{ED_prog_1}\\
    \text{s.t.} & \; \;
    \begin{bmatrix} F(u_{\text{d}})& U^{\top}\\ U & 0\end{bmatrix} \begin{bmatrix}g \\ \nu \end{bmatrix}=\begin{bmatrix} \hat{Y}_p(u_{\text{d}})^{\top}y_{\text{ini}} \\ \bar{u} \end{bmatrix}, \label{ED_prog_2}\\
& u_{\text{d}} \in \mathcal{U}, \label{ED_prog_3}
\end{align}
\end{subequations}
where: $\nu \in \mathbb{R}^{L}$ are Lagrangian multipliers; $\bar{u}=\left[ u_{\text{ini}}^\top \; \; u_{\text{s}}^\top \right]^\top$; and $F(u^d)=L \sigma^2 I_M + \hat{Y}_p(u_{\text{d}})^{\top} \hat{Y}_p(u_{\text{d}})$, where $M=c_H$ for Hankel matrices and $M=c_P$ for Page. The matrix $\hat{Y}_p$ is built with an output \emph{predicted} data trajectory $\hat{y}_{\text{d}}$ which linearly depends on the optimized input data trajectory $u_{\text{d}}$. This can be done using a baseline model, e.g. a previously identified impulse response model of the system. Analyses in \cite{Iannelli_SYSID_2021} showed that the accuracy of the baseline model has generally small impact on the input design problem. It is important to recognize that the trajectory optimized via (\ref{ED_prog}) automatically satisfies the input condition of Lemma \ref{DDsim_Relaxed}, since the constraint $Ug=\bar{u}$ effectively enforces the range constraint. As for the state condition of Lemma \ref{DDsim_Relaxed}, its fulfillment will depend on the choice of initial condition (see Remark 1). Finally, it is noted that, even though the Bayesian interpretation of problem (\ref{ED_prog}) only holds exactly for Page matrices, the same program will later be tested on Hankel matrices. This is done in order to preserve numerical tractability and is conceptually justified by the fact that the objective function (\ref{ED_prog_1}) can still be interpreted based on the A-optimality of the estimator's mean-square error (MSE) matrix \cite{Iannelli_SYSID_2021}.

\section{Numerical examples}\label{Results}
We consider the design of an input data trajectory $u_{\text{d}}$ for a data-driven simulation of the SISO system 
\begin{equation*}
    G(z) = \dfrac{0.1159(z^3+0.5z)}{z^4-2.2z^3+2.42z^2-1.87z+0.7225},
    \label{eqn:sys1}
\end{equation*}
which was originally investigated in \cite{Pillonetto_2010} and then also studied in \cite{Ming_SMM_Arxiv,Iannelli_SYSID_2021}. 
The accuracy of the estimated output simulation trajectory $\hat{y}_{\text{s}}$ is quantified via the following fit 
\begin{equation*} 
    W=100 \left(1-\left[\frac{\sum_{i=0}^{L_s-1}(y_{\text{s},i}-\hat{y}_{\text{s},i})^2}{\sum_{i=0}^{L_s-1}(y_{\text{s},i}-\bar{y}_{\text{s}})^2)^2}\right]^{1/2}\right),
\end{equation*}
where $\bar{y}_{\text{s}}$ is the mean of the true output sequence $y_{\text{s}}$.

For the input constraint set $\mathcal{U}$, a bound on the total energy is imposed here by defining $\mathcal{U} = \left\{u_{\text{d}} | \sum_{i=0}^{N-1}(u_{\text{d},i})^2 \leq E_0 N \right\}$. Other constraint sets (e.g. magnitude constraints) could be studied as well. The following parameters are kept fixed throughout the analyses:
$
L_0=4,L_s=10, L=14, E_0=0.1$. Note that $L_0$ is equal to the system order, thus the condition $L_0 \geq l$ is satisfied.
\begin{color}{black}
The baseline model consists of a truncated impulse response model of length $4 L_s$ estimated with the \emph{impulseest} MATLAB function using a prior experiment with i.i.d. Gaussian inputs of length $100$ with SNR=10. 
\end{color}
The solver IPOPT \cite{IPOPT} is employed to solve the nonlinear program (\ref{ED_prog}). The optimized input data trajectory is then used to do data-driven simulation with SMM, i.e.: compute $g$ by solving the relaxed version of (\ref{eqn:opt0}) proposed in \cite{Ming_SMM_Arxiv}; and then predict $\hat{y}_{\text{s}}$ with (\ref{SMM_sim}). 

In the first experiment, we compare the accuracy of Page and Hankel data matrices when data trajectories of \emph{same} length $N$ are used to predict the response to an impulse. Figure \ref{Fig_1_PvsH} shows the mean fit $W$ over 200 realizations of the noise as a function of $N$. Noise contaminates the output data trajectory as in (\ref{noise_data}), and two cases are considered, $\sigma^2=0.001$ and $\sigma^2=0.01$, which correspond respectively to a signal-to-noise ratio (SNR) of 100 and 10. 
\begin{figure}[h!]%
\begin{center}
\includegraphics[width=0.965\columnwidth]{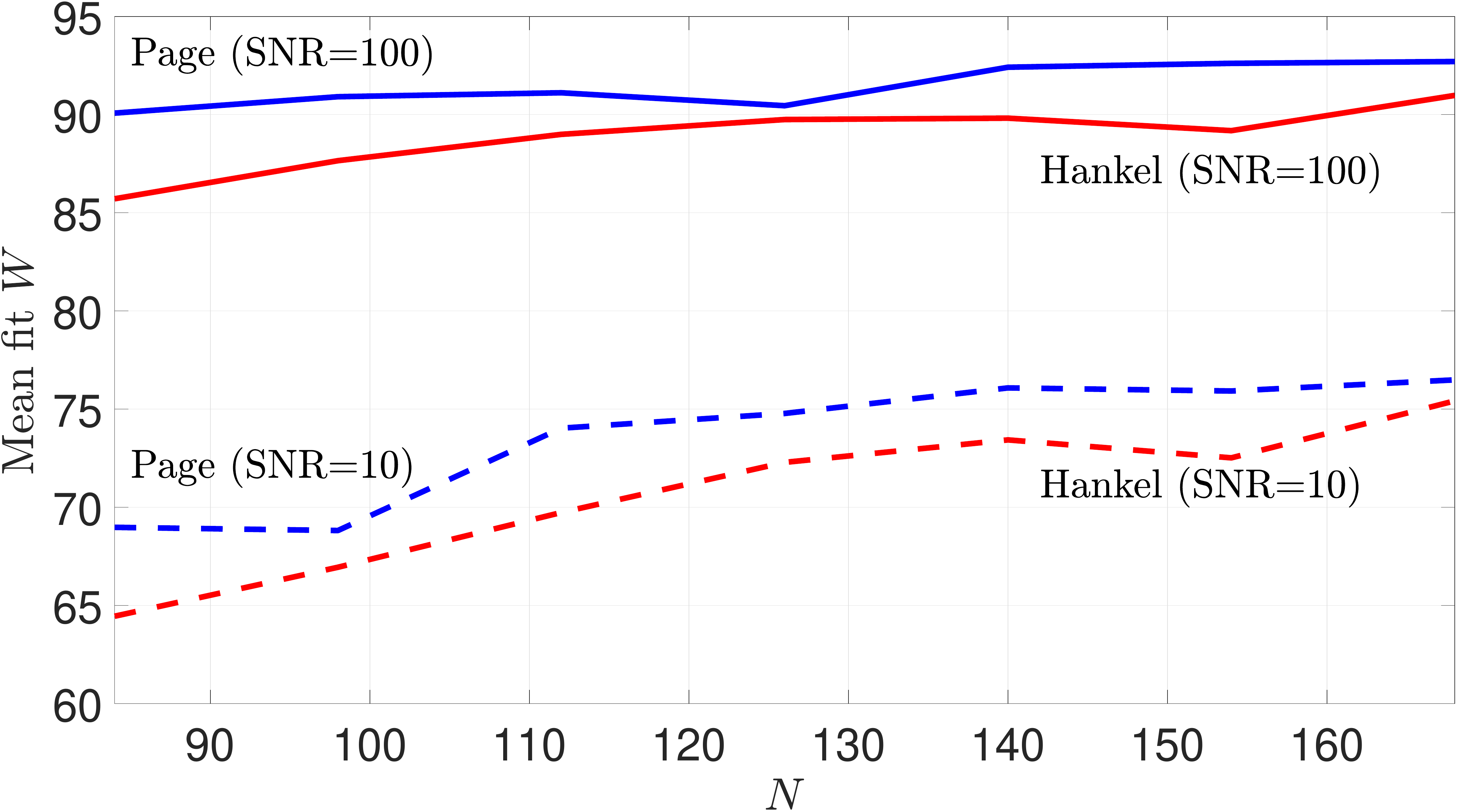}
\caption{Mean fit of Hankel (red) and Page (blue) matrices with optimally designed input as a function of $N$ and $\sigma^2$.}
\label{Fig_1_PvsH}
\end{center}
\end{figure}

The plot shows that Page matrices always outperform Hankel matrices. It is remarked that, since the comparison is made for equal length $N$ of the experiments, Hankel data matrices have a larger number of columns (precisely, $c_H(N)-c_P(N)=N(L-1)/L-L+1$ as discussed in Section \ref{DD_std}). In this type of comparison, typically done using i.i.d. Gaussian inputs, Page matrices are reportedly less accurate than Hankel matrices \cite{Markovsky2020identifiability}. The favourable trend reported in Figure \ref{Fig_1_PvsH} is achieved by using the relaxed excitation conditions from Lemma \ref{DDsim_Relaxed} together with the optimally designed input according to problem (\ref{ED_prog}). It is also noted that a comparison with standard persistently exciting inputs is not possible, for the Page matrix, because in the range of $N$ in Fig. \ref{Fig_1_PvsH} the classic excitation requirements of Lemma \ref{DDsim_Markovsky_Page} do not hold. 

In Figure \ref{Fig_2_boxPlots}, a comparison is made between the prediction accuracy of Page and Hankel matrices for two different simulation problems where $u_{\text{d}}$ is respectively a heavily and lightly damped sine wave.
\begin{figure}[h!]%
\begin{center}
\includegraphics[width=1\columnwidth]{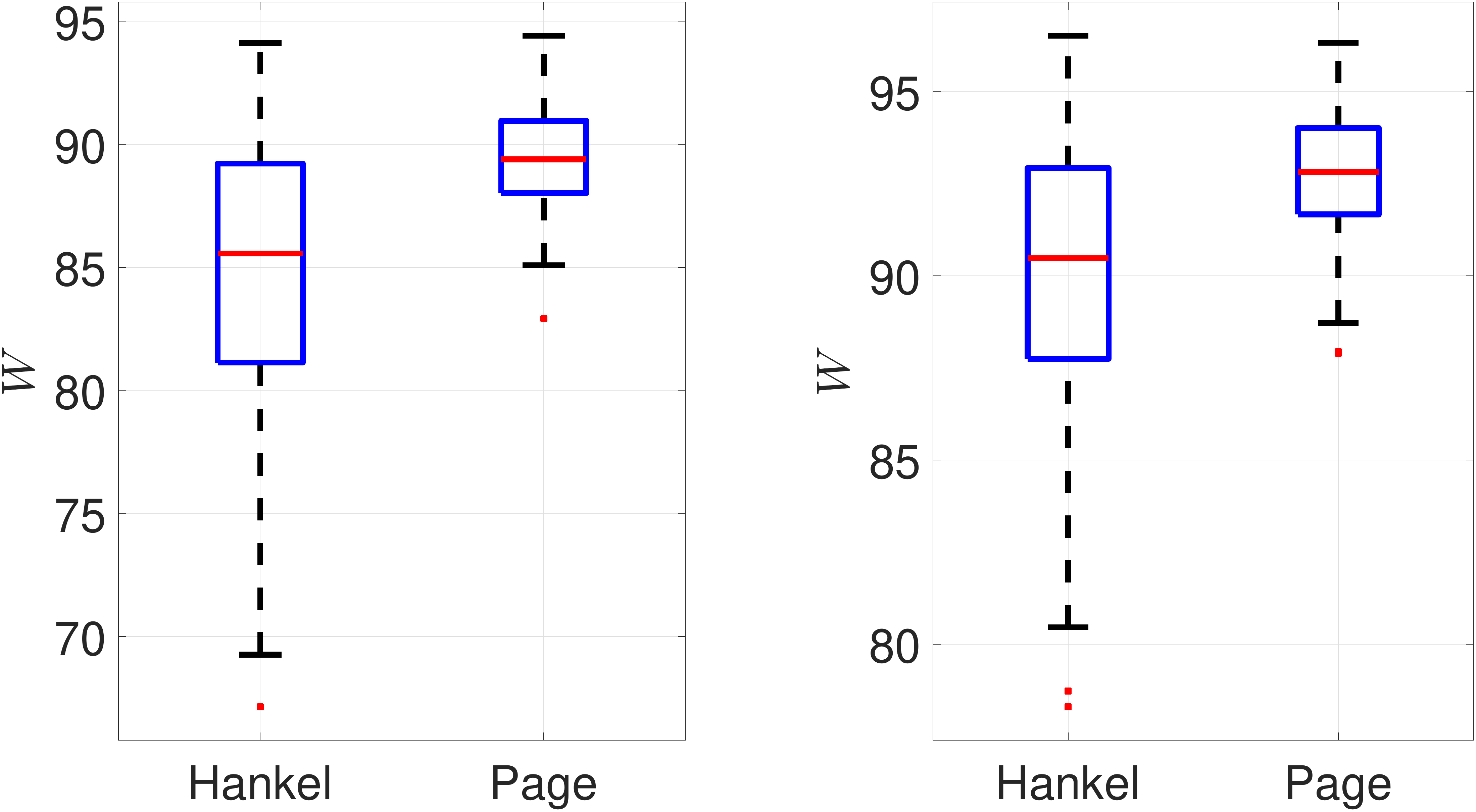}
\caption{Fit box plots for Page and Hankel matrices with $N=84$ and SNR=100 (200 realizations of the noise). Left: heavily, right: lightly, damped sine wave.}
\label{Fig_2_boxPlots}
\end{center}
\end{figure}

The advantage of employing Page matrices can be appreciated here also in terms of dispersion. There are a few aspects that can provide an explanation for these results. It is intuitively expected that input design is more effective with the Page matrix representation, as the matrix $U$ can be designed without structural constraints. In addition, Page matrices are known to have favourable properties when dealing with noisy data matrices \cite{Damen_Page_2982SCL}. Moreover, the input design criterion and the expression of the covariance $\Sigma_{\text{y}}$ used in the SMM problem are only exact for Page representations. All these reasons provide valuable justifications for the use of Page matrices in data-driven simulation problems. Lemma \ref{DDsim_Relaxed} is a key enabler that makes this possible in the regime of data length $N$ where a comparison with Hankel matrices is still meaningful.

\section{Conclusions}
Less restrictive excitation requirements have been presented for the data-driven simulation problem with clean data and Page and Hankel matrix representations. Building on these, a Bayesian input design problem for the case of noisy data has been formulated. For the Page matrix case, this can be interpreted as the choice of input trajectory maximizing the distance between prior and posterior distributions of the output response, while for the Hankel matrix this only holds approximately. Numerical results show that, by leveraging the results presented in the paper, the Page matrix representation can markedly outperform the Hankel one with the same data length. It is an interesting research question whether a similar approach to data informativity can be used when the data matrices are used for control.

\bibliographystyle{IEEEtran}

\bibliography{CDC21_exp}

\end{document}